\documentclass[11pt, letterpaper]{article}
\usepackage{amsmath,amssymb,amsfonts,amsthm}
\usepackage{algorithmic}
\usepackage{graphicx}
\usepackage{textcomp}
\usepackage{xcolor}
\usepackage[letterpaper,margin=1in]{geometry}
\usepackage[english]{babel}
\usepackage[numbers,sort&compress]{natbib}

\usepackage{amsmath,amssymb,array,color,colortbl,graphicx,multirow}
\usepackage{microtype}
\usepackage{enumitem}
\usepackage{algorithm2e}
\usepackage{listings,color}
\usepackage{verbatim}
\usepackage{multirow}
\usepackage{listings}
\usepackage{framed,color}
\usepackage{balance}
\usepackage{url}
\usepackage{hyperref}

\newtheorem{theorem}{Theorem}
\newtheorem{corollary}[theorem]{Corollary}
\newtheorem{lemma}[theorem]{Lemma}

\newcommand{\A}{\mathcal{A}}

\providecommand{\ie}{{i.e.,} }

\providecommand{\eg}{{e.g.,} }

\newcommand{\local}{\ensuremath{\mathsf{LOCAL}}}
\newcommand{\localp}[1]{\ensuremath{\mathsf{LOCAL}(#1)}}

\newcommand{\slocal}{\ensuremath{\mathsf{SLOCAL}}}
\newcommand{\pslocal}{\ensuremath{\mathsf{PSLOCAL}}}
\newcommand{\congest}{\ensuremath{\mathsf{CONGEST}}}

\newcommand{\supported}{\ensuremath{\mathsf{SUPPORTED}}}
\newcommand{\supportedp}[1]{\ensuremath{\mathsf{SUPPORTED}(#1)}}

\newcommand{\lcl}{\ensuremath{\mathsf{LCL}}}

\DeclareMathOperator{\poly}{poly}

\newenvironment{myabstract}
{\list{}{\listparindent 1.5em
		\itemindent    \listparindent
		\leftmargin    1cm
		\rightmargin   1cm
		\parsep        0pt}
	\item\relax}
{\endlist}

\newenvironment{mycover}
{\list{}{\listparindent 0pt
		\itemindent    \listparindent
		\leftmargin    1cm
		\rightmargin   1cm
		\parsep        0pt}
	\raggedright
	\item\relax}
{\endlist}

\newcommand{\myaff}[1]{\,$\cdot$\, {\small #1}\par\medskip}

\hypersetup{
	colorlinks=true,
	linkcolor=black,
	citecolor=black,
	filecolor=black,
	urlcolor=[rgb]{0,0.1,0.5},
	pdftitle={On the Power of Preprocessing in Decentralized Network Optimization},
	pdfauthor={Klaus-Tycho Foerster, Juho Hirvonen, Stefan Schmid, and Jukka Suomela}
}

\begin{document}

\begin{mycover}
	{\huge\bfseries\boldmath On the Power of Preprocessing in Decentralized Network Optimization \par}
	\bigskip
	\bigskip

	\textbf{Klaus-Tycho Foerster}
	\myaff{University of Vienna, Austria}

	\textbf{Juho Hirvonen}
	\myaff{Aalto University, Finland}

	\textbf{Stefan Schmid}
	\myaff{University of Vienna, Austria}

	\textbf{Jukka Suomela}
	\myaff{Aalto University, Finland}
\end{mycover}

\medskip
\begin{myabstract}
	\noindent\textbf{Abstract.}
As communication networks are growing at a fast pace,
the need for more scalable approaches to operate such
networks is pressing.
Decentralization and locality are key
concepts to provide scalability.
Existing models for which local algorithms are designed
fail to model an important
aspect of many modern communication networks
such as software-defined networks:
the possibility to precompute distributed network state.
We take this as an opportunity to study the
fundamental question of how and to what extent
local algorithms can benefit from preprocessing.
In particular, we show that preprocessing allows for significant speedups of various networking problems.
A main benefit is the precomputation of structural primitives, where purely distributed algorithms have to start from scratch.
Maybe surprisingly, we also show that there are strict limitations on how much preprocessing can help in different scenarios.
To this end, we provide approximation bounds for the maximum independent set problem---which however show that our obtained speedups are asymptotically optimal.
Even though we show that physical link failures in general hinder the power of preprocessing, we can still facilitate the precomputation of symmetry breaking processes to bypass various runtime barriers.
We believe that our model and results are of interest beyond the scope of this paper and apply to other dynamic networks as well.
\end{myabstract}

\vspace{5cm}
\begin{myabstract}
\noindent \copyright{} 2018 IEEE\@. This is the authors' version of a paper that will appear in the Proceedings of the IEEE International Conference on Computer Communications (INFOCOM 2019).
\end{myabstract}

\newpage

\section{Introduction}\label{sec:intro}

\subsection{Context: Decentralization for Scalability}

\emph{Locality}, the idea of avoiding  global collection of distributed network state and decentralizing the operation of networked systems,
is a fundamental design principle for scalability.
A more local operation cannot only reduce
communication overheads but also allow to
react faster to local events---an important aspect given
the increasingly stringent latency and dependability
requirements in future communication networks.
Given the quickly increasing scale of communication networks,
also due to the advent of new applications such as
the Internet-of-Things, the importance of
local approaches to networking is likely to
increase.

Designing local (\ie decentralized)
algorithms however can be challenging
and face a tradeoff: while a more
local network operation requires less
coordination (and
less overhead)
between neighboring domains,
hence improving scalability, a more limited
local view
may lead to suboptimal decision making:
compared to a global approach to network optimization,
a decentralized architecture may be subject to a
``price of locality''.

Local algorithms
have been studied intensively over the last decades,
and today, we have a fairly good understanding of
their opportunities and limitations~\cite{peleg00distributed,suomela13survey,DBLP:series/synthesis/2013Barenboim}.
Highly efficient distributed
algorithms are known for many network optimization
problems, sometimes even achieving an ``ideal
scalability'': their performance, \ie runtime,
is constant, \emph{independently} of the network size
(which could even be infinite)~\cite{suomela13survey}.
However, it is also known that for many fundamental network
optimization problems, \eg related to spanning tree~\cite{peleg2000near}
or shortest path~\cite{holzer2012optimal} computations, or minimizing congestion~\cite{keren-lowerbound},
designing good
decentralized optimization algorithms is impossible~\cite{kuhn2004cannot}:
in order to achieve non-trivial
approximations of the global optimum, non-local coordination
is required.

\subsection{Motivation: Improving Scalability with Preprocessing}

Our work is motivated by the observation
that existing models for the design of distributed
network algorithms, originally developed for ad-hoc
and sensor networks, do not account for
a key aspect of modern, large \emph{wired}
communication networks: the possibility
to \emph{preprocess} distributed network state.
In ad-hoc networks,
the network topology is typically assumed to be
unknown in the beginning,
and depend, \eg on the (unknown) node locations and
wireless communication channels.
The network topology hence
needs to be discovered in addition to performing
optimizations.
In contrast, most wired networks today have
a fairly static and known (to the operator)
network topology, which can hence be assumed
to be \emph{given}.
Network optimization algorithms here are mainly
concerned with the fast reaction
to new events, \eg finding optimal routes based on the
current traffic patterns.
While the topology of wired networks may change as well,
especially link additions but also link failures
happen less frequently and at different time scales.

This introduces an \emph{opportunity for preprocessing}
and calls for a radically new model:
decentralized and local algorithms reacting to
changes of the demand, network flows, or even failures,
may rely
on certain knowledge of the physical network topology,
based on which distributed
network state can be precomputed: this information
can later
``support'' local algorithms during their local
optimizations.

Indeed, enhancing classic, most basic
local coordination problems with preprocessing
appears to be a game changer:
distributed algorithms in traditional models
often require symmetry breaking mechanisms
whose complexity alone is in the same order
as solving the entire optimization problem~\cite{hotsdn13}.
With preprocessing, such symmetries
can trivially be broken \emph{ahead of time}.

\subsection{Case Study: Scalable SDNs}\label{subsec:SDNs}

As an example and case study, let us
consider emerging Software-Defined Networks (SDNs).
While there is now a wide consensus on the benefits
of moving toward more software-defined communication
networks, which are increasingly adopted
not only in datacenters~\cite{singh2015jupiter}
but also in wide-area
networks~\cite{jain2013b4,vahdat2015purpose}, one question becomes
increasingly pressing~\cite{feamster2013road}:
how to deploy SDNs \emph{at scale}?
In the near future, large-scale SDNs
may carry millions of flows and
span thousands of switches and
routers distributed across a large
geographic area.

A canonical solution to support large-scale deployments
while ensuring fast control plane reaction
to dataplane events (close to their origin),
is to partition the control plane~\cite{feamster2013road}
and \emph{leverage locality}:
different controller instances are made
responsible for a separate portion of the topology.
These controllers may then exchange e.g.
routing information with each other to ensure consistent
decisions. However,
SDN controllers face more general problems
than just routing, as they also need to support
advanced controller applications, including
involving distributed state
management.

Clearly, SDNs are very different from the models
usually studied in the context of ad-hoc networks:
an operator usually has full knowledge of the physical
connectivity of its networks, but needs clever algorithms
that allow distributed controllers to locally react
to new demands. In principle, these distributed controllers
can leverage such additional knowledge on the topology
and precompute distributed network state.

\subsection{Contributions}

Our work is motivated by the observation
that existing algorithmic locality models in
the literature are not a good fit
for modern large-scale wired networks such
as SDNs whose distributed control plane may
leverage precomputed network state.
Accordingly, we investigate a novel model
in which distributed decision making can
be supported by centralized preprocessing,
present and discuss different algorithmic techniques,
and derive lower bounds.

Indeed, we find several most fundamental problems
where significantly faster and hence more scalable
network algorithms
can be devised than in traditional models.
In particular, we show that preprocessing allows us to run many centralized algorithms efficiently in a distributed setting: \eg compared to existing work, we can achieve an almost exponential speedup for problems such as maximal independent set, as well as for approximation schemes of its optimization variant.
Furthermore, we can show that for locally checkable labellings, all symmetry breaking problems collapse to constant time complexity.

However, surprisingly, we also identify inherent limitations of the usefulness of precomputations; even seemingly simple problems still require linear runtimes.
Moreover, we also prove that our provided approximation schemes are essentially the best one can hope for: faster algorithms come at the price of worse approximation ratios.

Last in this list, we study the impact of physical link failures on the power of preprocessing. Unlike prior work, we assume that the link failures can be arbitrary, \ie only a small part of the resulting physical topology could remain, possibly in many disconnected parts.
While we formally prove that this model inherently limits the general power of precomputation, we can still speed up various distributed algorithms for restricted graph classes respectively bounded node degrees.

\subsection{Organization}

The remainder of this paper is organized
as follows.
In Section~\ref{sec:model}, we begin our investigation by providing a formal model that captures the capabilities of preprocessing, along with defining some key notation.
Section~\ref{sec:examples} then provides two introductory examples to outline the power respectively limitations of preprocessing.
We next show in Section~\ref{sec:SLOCAL} how a recently introduced problem class of central computation can be turned distributed with small overhead, by leveraging preprocessing.
Section~\ref{sec:LCL} investigates the power of
preprocessing in an important class of locally verifiable problems, and Section~\ref{sec:opt} studies approximation of maximum independent sets.
We further prove lower bounds in Section~\ref{sec:lower-bound}, providing matching bounds for the algorithmic results of the previous section.
Related work is discussed in Section~\ref{sec:relwork}, followed by concluding remarks in Section~\ref{sec:conclusion}.

\section{Model}\label{sec:model}

Let us first revisit the
$\local$ model which is used predominantly
today to design and analyze distributed
network algorithms~\cite{peleg00distributed}.
In the $\local$ model the communication structure is given by an undirected graph
$G=(V,E)$,
 where each node
$v \in V$
is a computer with a unique identifier.
The nodes can exchange messages between neighbors across edges in each communication round, where in parallel each node can send and receive a message to/from each neighbor, and afterward update its local state.
Before the first communication round, each node is provided with an input that defines the problem instance.
The running time of an algorithm is the number of rounds needed until all nodes stopped, i.e., they announce a final output and no longer communicate.

In this paper, we are interested in the benefits of extending the $\local$ model with an opportunity of preprocessing the network topology.
We will refer to our model as the
$\supported$ model~\cite{hotsdn13}.
In the $\supported$ model the communication structure is given by an undirected connected graph $H=(V(H),E(H))$ (the \emph{support} graph), where again each node\footnote{In the context of SDNs (Section~\ref{subsec:SDNs}), each node in the $\supported$ model could correspond to a controller domain.}
 is a computer with a unique identifier.
However, the problem instance (the \emph{logical} state) is defined on a subgraph $G \subseteq H$, where $G$ is called the \emph{input} graph $G=(V,E)$ and inherits the node identifiers of $H$.
Computation in the $\supported$ model now proceeds in two steps:
$(1)$ As a preprocessing phase, the nodes may compute any function depending on $H$ and store their output in local memory.
Next, $(2)$ the nodes are tasked to solve a problem on the input graph $G$ in the $\local$ model, where additionally the preprocessing from $(1)$ and the communication infrastructure of $H$ may be used.
The running time $t$ of an algorithm in the $\supported$ model is measured in the number of rounds needed in step~$(2)$, also denoted as $\supported(t)$.
We will use analogous notation also for other models, e.g., $\local(t)$.

\section{Potential and Limitations}\label{sec:examples}

To provide intuition, we first present an example
that shows the potential of the $\supported$ model
to overcome inherent limitations of the $\local$ model.

\paragraph{Fast symmetry breaking.}
One of the most important problems~\cite{peleg00distributed} in the $\local$ model is \emph{coloring}, \ie providing each node with a color different to its neighbors', with the goal of using a small color palette.\footnote{For a fast (inefficient) coloring solution, each node could assign its unique identifier, \eg its MAC-address, as its color, requiring $n$ colors in total.} Such symmetry breaking is often the first step for more advanced problems, but also has direct applications, such as, \eg in scheduling. Even in very simple topologies such as a ring, computing a 3-coloring in the $\local$ model takes non-constant time~\cite{linial92locality}, and 2-coloring even requires $\Omega(n)$ rounds.
By using the power of precomputation, we can assign an efficient coloring based on the support graph ahead of time, which is then still valid on the input graph.
As such, the $\supported$ model can significantly speed up the computation of many $\local$ algorithms if the support graph allows for a coloring with few colors.
We note that this is just an introductory example, we cover more complex cases starting in the following Section~\ref{sec:SLOCAL}.

\paragraph{Not everything can be precomputed.}
However, there are also inherent limitations.
For example in \emph{leader election}, the task is to provide each connected component with a single coordinator.
In the $\local$ model, a leader can be elected in time equivalent to the diameter of such a component, assuming identifiers are present~\cite{DBLP:conf/stoc/Angluin80}.
From a theoretical point of view, the $\supported$ model cannot improve this time bound: as the input graph can be any subgraph of the support graph, nodes must first discover their connected component. If no leader exists, a new one must be elected, and if multiple leaders exist, they must coordinate with each other to declare a unique leader, possibly over long distances.
For an extreme example, consider a line as a support topology, where the input graph can consist of many disconnected components of large sizes.


\section{Exploiting the Locality of Global Problems}\label{sec:SLOCAL}
\paragraph{The $\slocal$ model.}
Many networking problems are slow to solve from scratch in a distributed setting.
For example, the best known algorithm for maximal independent set requires
$2^{O(\sqrt{ \log n})}$
rounds in the \local{} model~\cite{DBLP:journals/jal/PanconesiS96}, but the problem is rather trivial in a centralized setting:
go through the nodes sequentially, adding them to the independent set if none of their neighbors are in it.

The maximal independent set problem has a low locality (just the node's directed neighbors are important), but still defies fast distributed algorithms.
Recent work formalized this notion of locality from a globalized point of view, coined the \slocal{} model~\cite{DBLP:conf/stoc/GhaffariKM17}.
In the \slocal{} model, the locality $t$ is defined as the radius $t$-neighborhood a node may use to compute its output---where the nodes are processed sequentially in an arbitrary order.
For example, the maximal independent set problem is therefore in \slocal{} with locality $t=1$.

\paragraph{An $\slocal$ simulation.}

Ghaffari et al.\ \cite{DBLP:conf/stoc/GhaffariKM17} provided an approach to transforming any algorithm from the \slocal{} model into a purely distributed algorithm by simulating it in the \local{} model. This is done by computing a locally simulatable execution order for the \slocal{} algorithm: each node $v$ must know on what other nodes' output does $v$ need to wait on before computing its output?

We can adapt their techniques to also apply to the \supported{} model:

\begin{theorem}\label{thm:slocal-sim}
  An \slocal{} algorithm with locality $t$ can be simulated in the
	\supported{} model in time $O(t \cdot \poly \log n)$.
\end{theorem}

\begin{proof}	Ghaffari et al.\ \cite{DBLP:conf/stoc/GhaffariKM17} show how network decompositions can be leveraged to compute dependencies of bounded depth:
	as some graphs, \eg the complete graph, require dependency chains of length $\Omega(n)$, it can be useful to also incorporate communication between the nodes.
	The idea is that the nodes learn all their nearby dependency chains. For example, in the complete graph, one round suffices to gather all information.
	With this modification, every graph has an ordering of depth $O(\poly \log n)$, if the nodes are allowed to communicate within a distance of $O(\poly \log n)$~\cite{DBLP:conf/stoc/GhaffariKM17}.

	We start with the case of $G=H$.
	After providing such an ordering, every \slocal{} algorithm with locality $t$ can be run in time $O(t \cdot \poly \log n)$ in the \local{} model, by simulating executions of lower priority in polylogarithmic distance.
	We can provide the ordering in the \supported{} model in the precomputation phase, finishing the case of $G=H$.

	We next cover the remaining case of $G\neq H$.
	Observe that a dependency chain is still valid even if edges are deleted, as nodes just need to respect the computation of all neighbors of lower priority --- but under edge deletions, a node might need longer to learn about previously nearby dependency chains.
	However, in the \supported{} model, the edges are not physically deleted\footnote{We study this model variant in \S\ref{sec:PASSIVE}, coined  \emph{passive} \supported{} model.}, but rather just not part of the problem graph anymore, still available for communication purposes.
	The simulation approach therefore still holds for $G\neq H$.
\end{proof}
Hence, we can obtain a maximum independent set in time $O(\poly \log n)$ in the \supported{} model, as it had an \slocal{} locality of $1$. Ghaffari et al.\ \cite{DBLP:conf/stoc/GhaffariKM17} define the class \pslocal{} as \slocal{}($\poly \log n$).
We obtain the following corollary:

\begin{corollary}
 $\pslocal{} \subseteq \supportedp{\poly \log n}$.
\end{corollary}

\paragraph{Separation of $\slocal$ and $\supported$ model.}
As we have seen, the \slocal{} model can be simulated in the
 \supported{} model with polylogarithmic overhead.
Under some additional assumptions, the converse it not true:
\begin{itemize}
	\item \emph{Degree/colorability restrictions of the support graph $H$}: \newline If we restrict the support graph $H$ to be 2-colorable (or to be of maximum degree 2), then we can precompute a 2-coloring which is also valid for any subgraph $G \subseteq H$.
	Computing a 2-coloring on $G$ in the $\slocal$ model however has a locality of $\Omega(n)$.
	\item \emph{Unknown range of IDs / network size}: Computing an upper bound on the network size requires a locality of $\Omega(n)$ in the \slocal{} model, but is trivial with preprocessing.
	\item \emph{Inputs already known in the preprocessing phase}:\newline Similarly, if some problem inputs are already known ahead of time in the $\supported$ model, large speedups can be obtained in comparison to the $\slocal$ model.
\end{itemize}
However, we leave it as an open question if there is a strict (\eg beyond polylogarithmic) separation between the $\slocal$ and $\supported$ model without the above assumptions.


\section{Locally Checkable Labellings} \label{sec:LCL}

\emph{Locally checkable labellings}, first introduced by Naor and Stockmeyer~\cite{naor95what}, are a family of problems that can be efficiently verified in a distributed setting. They include fundamental problems like maximal independent set, maximal matching, $(\Delta+1)$-coloring, and $(2\Delta-1)$-edge coloring.

More formally, let $\Sigma$ and $\Gamma$ denote a finite sets of input and output labels, respectively. A graph problem consists of a set of labelled graphs with some constant maximum degree $\Delta$. A problem $\Pi$ is locally checkable if there exists a constant $r$ and a distributed algorithm $A$ with running time $r$ such that given a labelling $\lambda \colon V(G) \rightarrow (\Sigma \times \Gamma)$, we have that $A$ outputs \emph{yes} at every node if and only if $(G,\lambda) \in \Pi$. In particular, if $(G,\lambda) \notin \Pi$, at least one node outputs \emph{no}.

\begin{theorem} \label{thm:lcl-collapse}
  Every \lcl{} in \localp{o(\log n)} can be solved in time $O(1)$ in the \supported{} model.
\end{theorem}

The proof follows the idea of Chang et al.\ \cite{DBLP:conf/focs/ChangKP16}: nodes compute a coloring that locally looks like a unique identifier assignment, but from a smaller space. Then we simulate a \local{} algorithm on this identifier setting and it must terminate fast.

\begin{proof}
  Let $A$ be an algorithm for an $\lcl{}$ $P$ with running time $T(n) = o(\log_{\Delta(G)} n)$ in the \local{} model (here $T(n)$ might include functions of $\Delta(G)$, \ie of the maximum degree of the input graph, but these are constants). Without loss of generality we assume that $P$ can be checked with radius 1 (this is a standard transformation---nodes encode also the outputs of their neighborhood as necessary).

  Given an input graph $G$ with maximum degree $\Delta = \Delta(G)$, the ball of radius $t$ around each node holds at most $1 + \Delta(\Delta-1)^{t-1}$ nodes. We want to find the smallest $n_0$ such that any nodes of $G$ at distance at most $2T(n_0)+2$ can be colored with different colors from $[n_0]$. By definition of $T(n)$ we can find a constant $n_0$ such that $\Delta^{2T(n_0)+2}+1 \leq n_0$, that is, we can greedily color nodes with locally unique colors.

  Given a support $H$ of size $N$, compute an $n_0$-coloring $\phi$ of distance $(2T(n_0)+1)$ of $H$. This can be done without communication given the support. Now we apply $A$ with the coloring $\phi$ on the input graph $G$. This can be done in time $T(n_0) = O(1)$. As $G$ is a subgraph of $H$, the coloring $\phi$ is still a proper distance $2T(n_0)+2$-coloring of $G$. Since each node sees only values from $[n_0]$ inside its $T(n_0)$-neighborhood and each value is unique, the algorithm is well defined on $G$ with $\phi$ representing the identifiers of the nodes. Finally, since locally in each 1-neighborhood the output of the nodes looks identical to some execution in a graph of size $n_0$, the output must be consistent with the legal outputs of $P$. Since $P$ is an \lcl{}, an output that is correct everywhere locally is correct also globally.
\end{proof}

In the standard \local{} model locally checkable labellings can have three types of complexities: trivial problems solvable in constant time, symmetry breaking problems with complexity $\Theta(\log^* n)$, and problems with complexity $\Omega(\log n)$. Theorem~\ref{thm:lcl-collapse} implies that all symmetry breaking problems collapse to constant time in the \supported{} model.

We are unable to establish the existence of ``intermediate'' \lcl{} problems for the \supported{} model \cite{brandt2016LLL,DBLP:conf/focs/ChangKP16,ghaffari2017splitting}, in particular a lower bound of $\Omega(\log n)$ for an \lcl{} in $\localp{\poly \log n}$. We do show that there are problems that are hard despite the power of the support.

Consider the problem of sinkless orientation: each edge must be oriented so that every node has an outgoing edge. Usually this problem is defined so that nodes with degree 2 or less can be sinks. We consider the variant where every node with degree at least 2 must not be a sink.

\begin{theorem} \label{thm:so-lower-bound}
  Finding a sinkless orientation requires $\Omega(n)$ time in the \supported{} model.
\end{theorem}

\begin{proof}
  Consider the following support $H$ on $N = 6n$ nodes. Let $P_i = (v_{i,1}, v_{i,2}, \dots, v_{i,n})$, for $i \in \{1,2,\dots, 6\}$, denote six paths of length $n$. Add edges $\{v_{1,1}, v_{2,1}\}, \{ v_{1,1},v_{3,1} \}, \{ v_{1,n}, v_{4,1} \}, \{v_{1,n}, v_{5,1} \}$ and edges $\{v_{6,1}, v_{2,n}\}, \{ v_{6,1},v_{3,n} \}, \{ v_{6,n}, v_{4,n} \}, \{v_{6,n}, v_{5,n} \}$ between these paths.

  Observe that on a cycle a sinkless orientation is a \emph{consistent} orientation: each node must have exactly one incoming and one outgoing edge.

  Now consider the cycles $C_1 = (v_{1,1}, P_2, v_{6,1}, P_3)$ and $C_2 = (v_{1,n}, P_4, v_{6,n}, P_4)$ and assume that these are maximal connected components of $G$. Now both cycles must be oriented consistently in any valid solution. Without loss of generality assume that an algorithm $A$ orients path $P_2$ from $v_{1,1}$ to $v_{6,1}$, path $P_3$ from $v_{6,1}$ to $v_{1,1}$, path $P_4$ from $v_{1,n}$ to $v_{6,n}$, and path $P_5$ from $v_{6,n}$ to $v_{1,n}$.

  Now consider a graph $G'$ where the cycle $(P_2, P_1, P_4, P_6)$ forms a maximal connected component. If paths $P_2$ and $P_4$ are oriented as in $G$, then the solution cannot form a consistent orientation. Therefore either node $v_{2,\lceil n/2 \rceil}$ or node $v_{4,\lceil n/2 \rceil}$ must change its output between $G$ and $G'$. Since distinguishing between $G$ and $G'$ requires $\lceil n/2 \rceil$ rounds for these nodes, we conclude that solving sinkless orientation in the \supported{} model requires $\Omega(n)$ rounds.
\end{proof}


\section{Maximum Independent Set}\label{sec:opt}

We show how to find large independent sets in logarithmic time in the \supported{} model. This leads to an approximation scheme in graphs that have independent sets of linear size.

Let $\alpha(G)$ denote the fraction of nodes in the maximum independent set of graph $G$.

\begin{theorem} \label{thm:passive-lco}
  There is an algorithm that finds an independent set of size $(\alpha(G) - \epsilon)n$, for any $\epsilon > 0$, in time $O(\log_{1+\varepsilon} n)$ in the \supported{} model.
\end{theorem}

\begin{proof}
  To prove this we use a standard ball growing argument~\cite{linial94low,DBLP:conf/stoc/GhaffariKM17} which essentially states that a graph can only expand for a logarithmic number of hops. We denote by $B_t(v)$ the $t$-hop neighborhood of node $v$. Let $H$ be a graph: for any $\varepsilon > 0$ and each node $v \in V(H)$, there exists a radius $r \leq \log_{1+\varepsilon} n$ such that $|B_{r+1}(v)| < (1+\varepsilon)|B_r(v)|$. To see this, assume the contrary: since $|B_{i+1}(v)| \geq (1+\varepsilon)|B_i(v)|$, for $i$ up to $r > \log_{1+\varepsilon} n$, we have that $|B_{r}(v)| \geq (1+\varepsilon)^r > (1+\varepsilon)^{\log_{1+\varepsilon} n} = n$, a contradiction.

  In the preprocessing phase we decompose the graph into parts that have logarithmic diameter and a small boundary. Then the problem can be solved optimally in each such component and fixed on the boundaries.

  Consider an arbitrary node $v_1 \in V(H)$ and find the smallest radius $r_1$ such that $|B_{r+1}(v_1)| < (1+\varepsilon)|B_r(v_1)|$ as observed previously. Define $C_1 = B_{r+1}(v_1)$ and remove $C_1$ from $H$. By definition, at most a fraction of $\varepsilon' = \varepsilon / (1+\varepsilon)$ nodes in $B_{r+1}(v_1)$ are connected to $G \setminus B_{r+1}(v_1)$.
  Continue by selecting another node $v_2$ and finding the smallest $r_2$ such that the boundary of $B_{r+1}(v_2)$ has an $\varepsilon'$-fraction of the nodes. Again put these nodes into cluster $C_2$ and remove them from $H \setminus C_1$. Proceed in this fashion until all nodes (neighborhoods) have been allocated to a cluster.

  Since at each step each cluster has an $\varepsilon'$-fraction of nodes on the boundary (in $B_{r+1}(v) \setminus B_r(v)$), we have at least a $(1-\varepsilon')$-fraction of nodes strictly inside the clusters (that is, not on the boundary). Denote by $\bar{C_i}$ the inner nodes of cluster~$C_i$.

  Now given the input graph $G$, each node $v$ in cluster $C_i$ gathers the subgraph of $G$ induced by $C_i$. Then all nodes use the same algorithm to find an optimal maximum independent set of $C_i$, and put themselves in the (independent) set $I$ if they are included in this solution. This can be done using local computation, since the subgraph is known to all nodes.

  Finally, there are at most $\varepsilon' n$ nodes that are connected to two clusters. In the worst case all of these are adjacent to another node in the chosen set $I$. By removing each such node that has a neighbor in $I$ with a smaller identifier we make $I$ independent. Let $I^*$ denote a maximum independent set of $G$ and $I[S]$ a set $I$ restricted to a subgraph $S$. Since $I^*[C_i]$ is an independent set of $C_i$, we have that $|I[C_i]| \geq |I^*[C_i]|$ for all $C_i$. Therefore $I$ has size at least $|I^*| - \varepsilon'n$.
\end{proof}

Now if the family of graphs from which the support is drawn has a linear lower bound on the size of the maximum independent set, the solution $I$ constitutes a $(1+\varepsilon)$-approximation of~$I^*$.

\begin{corollary} \label{cor:passive-lco}
  Maximum independent set can be approximated to within factor $(1+\epsilon)$, for any $\epsilon > 0$, on graphs of constant maximum degree in time $O(\log_{1+\varepsilon} n)$ in the \supported{} model.
\end{corollary}

Note that the proof generalizes to a larger family of optimization problems: essentially we required that the global score of a solution is a sum of local scores, and that the nodes of the boundary constitute only a small fraction of the weight of the full solution. As an example, maximum cut and maximum matching are not problems of this type, as the input graph could only contain edges on the boundaries of the clusters, making the precomputation unhelpful.

The above results have the best possible running time as a function of $n$: finding an independent set of size $\omega(\frac{n \log \Delta}{\Delta})$ requires logarithmic time, as will be shown in Section~\ref{sec:lower-bound}. Ghaffari et al.\ \cite{DBLP:conf/stoc/GhaffariKM17} used the same ball growing idea to give \slocal{} algorithms for computing $(1+\varepsilon)$-approximations of covering and packing ILPs in time $O(\log n / \varepsilon)$. Their algorithm can be simulated via Theorem~\ref{thm:slocal-sim} in time $O(\poly \log n)$.

Note that we are abusing the unlimited computational power of the \supported{} model: solving each cluster optimally is an \textsf{NP}-hard problem.

\section{Lower Bounds}\label{sec:lower-bound}

In the previous section, we saw that maximum independent sets can be approximated well in the $\supported$ model in \emph{logarithmic} time. To complement this result, we now give a lower bound that shows that any \emph{sublogarithmic}-time algorithm in the $\supported$ model necessarily results in a poor approximation ratio:

\begin{theorem}\label{thm:mislb}
  Maximum independent set cannot be approximated by a factor $o(\Delta/\log \Delta)$ in time $o(\log_{\Delta} n)$ in the $\supported$ model.
\end{theorem}

This bound is tight in the sense that it can be matched in triangle-free regular graphs: Shearer~\cite{shearer83triangle} noted that the randomized greedy algorithm finds an independent set of size $O(\frac{\ln d}{d}n)$, and this can be approximated with a small loss in randomized constant time in the \local{} model, for example using the method of random priorities due to Nguyen and Onak~\cite{nguyen08constant}.

More precisely, we analyze the \emph{expected} approximation ratio here. To prove the result, assume that there is a family of (randomized) algorithms $\A_d$ such that for each even $d$ algorithm $\A_d$ finds a factor $\alpha_d$ approximation of a maximum independent set in any graph for which we have a $d$-regular support. Furthermore, assume that $\A_d$ runs in time $T(n,d) = o(\log_d n)$ in the $\supported$ model. To reach a contradiction, assume that $\alpha_d = o(d/\log d)$.

Fix a $d = 2k$. Using the probabilistic method, we can find a sufficiently large $n_Q$ such that there exists a $k$-regular graph $Q = (V_Q,E_Q)$ with $n_Q$ nodes with the following properties \cite{frieze92independence,wormald99regular,alon10constant}:
\begin{enumerate}[noitemsep]
    \item the girth of $Q$ is larger than $2T(2n_Q, d) + 1$,
    \item there is no independent set with at least $n_Q/(2\alpha_d)$ nodes.
\end{enumerate}

We will consider graphs with $n = 2n_Q$ nodes: for each original node $v \in V_Q$, we will have two copies $v_0$ and $v_1$. Let
\[
    V = \bigl\{ v_x : v \in V_Q, x \in \{0,1\} \bigr\}
\]
be the set of all copies. For each $F \subseteq E_Q$ we define
\[
\begin{split}
    E(F) &= \bigl\{ \{u_x,v_{1-x}\} : \{u,v\} \in F, x \in \{0,1\} \bigr\} \cup {} \\
           &\phantom{{}={}} \bigl\{ \{u_x,v_x\} : \{u,v\} \in E_Q \setminus F, x \in \{0,1\} \bigr\}.
\end{split}
\]
For each $F$, the graph $G(F) = (V, E(F))$ is a double cover of $Q$; intuitively, $F$ tells which pair of edges goes ``straight'' and which goes ``across''. In particular, we are interested in the following graphs (see Figure~\ref{fig:lbsupp}):
\[
    H = (V, E(\emptyset) \cup E(E_Q)), \quad
    G_1 = G(\emptyset), \quad
    G_2 = G(E).
\]
Here graph $G_1 \cong 2Q$ consists of two copies of $Q$ and graph $G_2 \cong Q \times K_2$ is the bipartite double cover of $Q$. By construction, there is no maximum independent set in $G_1$ with at least $n/(2\alpha_d) = n_Q/\alpha_d$ nodes. However, the largest independent set of $G_2$ has at least $n/2 = n_Q$ nodes, as $\{v_0: v \in V_Q\}$ is an independent set.

\begin{figure}[t]
    \centering
    \includegraphics[page=1, width=0.7\textwidth]{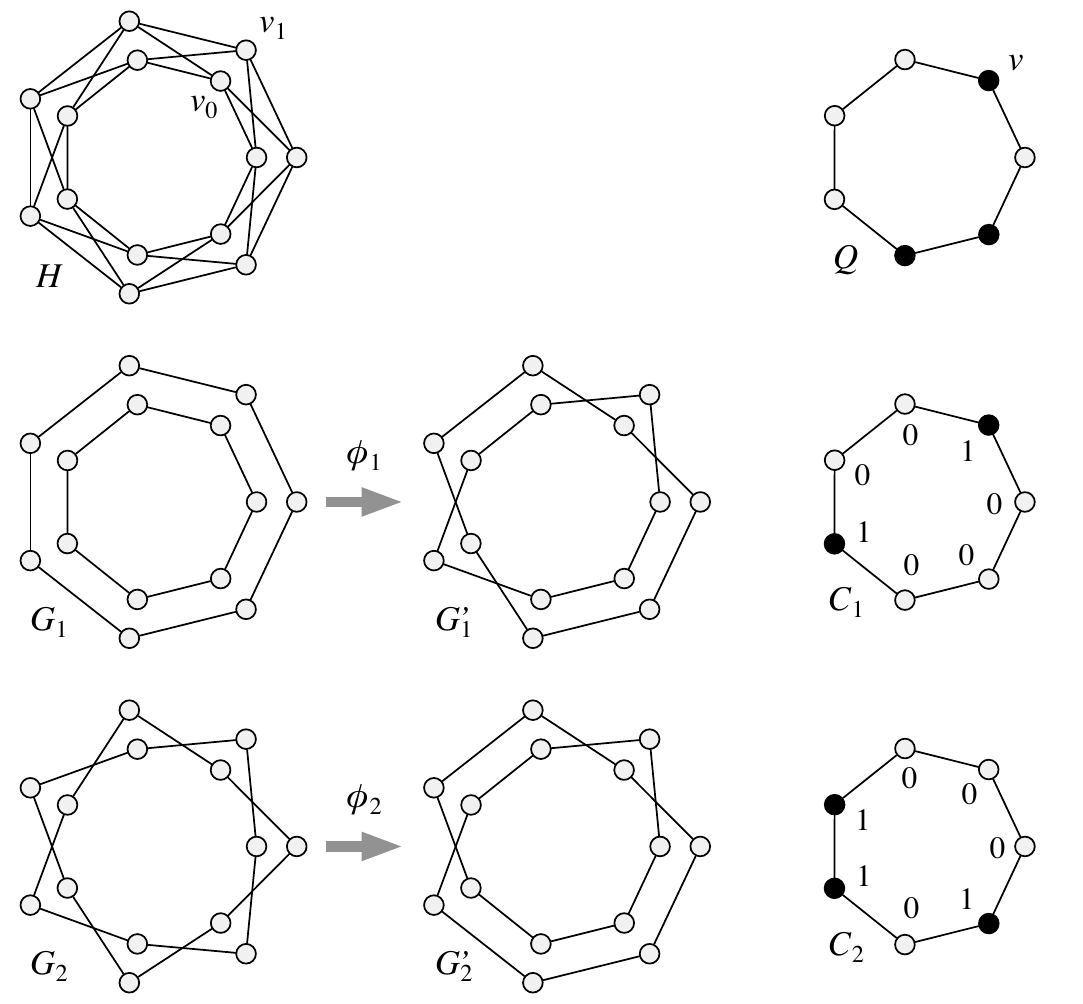}
    \caption{Construction of graphs $H$, $G_1$, $G_2$, and random graphs $G_1'$ and $G_2'$. In this example, $k = 2$ and graph $Q$ is a $7$-cycle. No matter how we choose $C_1$, graph $G_1'$ consists of two disjoint copies of $Q$, while $G_2'$ is a bipartite double cover of $Q$ (which happens to be a $14$-cycle). While random graph $G_1'$ and $G_2'$ have a different structure from a global perspective, they are locally indistinguishable in the sense that each possible local neighborhood occurs in both graphs with the same probability.}\label{fig:lbsupp}
\end{figure}

Now we will construct random graphs $G_1'$ and $G_2'$ as follows (see Figure~\ref{fig:lbsupp}). For each $i \in \{1,2\}$ and $v \in V_Q$, choose a label $C_i(v) \in \{0,1\}$ uniformly at random. Then let $X_i \subseteq E_Q$ be the set of edge $\{u,v\} \in E_Q$ with $C_i(u) \ne C_i(v)$. That is, $C_i$ is a uniform random cut of $Q$, and $X_i$ is the set of cut edges w.r.t.\ $C_i$. Define
\begin{align*}
    G_1' &= G(F_1), \text{ where } F_1 = X_1 \text{ and } \\
    G_2' &= G(F_2), \text{ where } F_2 = E_Q \setminus X_2.
\end{align*}

First, we will argue that from a global perspective, graphs $G_1'$ and $G_2'$ are different.

\begin{lemma}
    Graph $G_1'$ is isomorphic to $G_1$ and graph $G_2'$ is isomorphic to $G_2$.
\end{lemma}
\begin{proof}
For each cut $C_i$ of graph $Q$, define the bijection $\phi_i\colon V \to V$ as follows:
    \[
        \phi_i(v_x) = \begin{cases}
            v_x & \text{if } C_i(v) = 0, \\
            v_{1-x} & \text{if } C_i(v) = 1.
        \end{cases}
    \]
    That is, $\phi_i$ exchanges the labels of the copies of $v \in V$ whenever $C_i(v) = 1$. In particular, $\phi_i$ exchanges labels at exactly one endpoint of $e \in E_Q$ iff $e$ is a cut edge w.r.t.\ $C_i$. By construction, we have
    \begin{align*}
        E(F_1) &= \bigl\{ \{ \phi_1(u), \phi_1(v) \} : \{u,v\} \in E(\emptyset) \bigr\}, \\
        E(F_2) &= \bigl\{ \{ \phi_2(u), \phi_2(v) \} : \{u,v\} \in E(E_Q) \bigr\}.
    \end{align*}
    That is, $\phi_1$ is a graph isomorphism from $G_1$ to $G_1'$, and similarly $\phi_2$ is a graph isomorphism from $G_2$ to $G_2'$.
\end{proof}

\begin{corollary}\label{cor:G1mis}
    The maximum independent set of $G_1'$ has fewer than $n_Q/\alpha_d$ nodes and the maximum independent set of $G_2'$ has at least $n_Q$ nodes.
\end{corollary}

However, while the random graphs $G_1'$ and $G_2'$ are globally very different, we will argue that they are locally indistinguishable from the perspective of $\supported$ algorithms, if we use the same support $H$.

Let $T = T(2n_Q, d) = T(n,d)$ be the running time of algorithm $\A_d$ when we run it in the random graph $G_i'$ with support $H$. Consider a node $u \in V_Q$ and one of its copies, say, $u_0 \in V$.

As $Q$ is a graph of girth larger than $2T+1$, the radius-$T$ neighborhood of $u$ is a regular tree. Let $V^T_Q \subseteq V_Q$ be the set of nodes in the tree. Let $H^T$ be the radius-$T$ neighborhood of $u_0$ in the support $H$. Note that $H^T$ is the subgraph of $H$ induced by the copies of the nodes in $V^T_Q$. As the running time of $\A_d$ is $T$, the probability distribution of the output of node $u$ only depends on the input within $H^T$.

Construct a labelling $D\colon V^T_Q \to \{0,1\}$ that labels the nodes of the tree by their distance from $u$ modulo~$2$, and define a bijection $f\colon C_1 \mapsto C_2$ between labellings $C_1$ and $C_2$ as follows. Given any labelling $C_1\colon V_Q \to \{0,1\}$, define a labelling $C_2\colon V_Q \to \{0,1\}$ by setting $C_2(v) = C_1(v) + D(v) \bmod 2$ for each $v \in V^T_Q$ close to $u$, and let $C_2(v) = C_1(v)$ for each $v \notin V^T_Q$ far from $u$.

This is clearly a bijection. In particular, the following random processes are indistinguishable: (1)~pick $C_2$ uniformly at random; (2)~pick $C_1$ uniformly at random and let $C_2 = f(C_1)$.

Recall that we used $C_1$ to construct graph $G_1'$ and $C_2$ to construct graph $G_2'$. The key observation is this: if we set $C_2 = f(C_1)$, then the structure of $G_1'$ restricted to $H^T$ is isomorphic to the structure of $G_2'$ restricted to $H^T$. In particular, the probability that $u_0$ joins the independent set in $G_1'$ is the same as the probability that $u_0$ joins the independent set in~$G_2'$.

Summing over all choices of $C_1$ (and hence summing over all choices of $C_2 = f(C_1)$), we conclude that the probability that $u_0$ joins the independent set is the same in random graphs $G_1'$ and $G_2'$. Summing over all choices of $u_0$ (and similarly $u_1$), we see that the expected size of the independent set produced by $\A_d$ is the same in $G_1'$ and $G_2'$.

By Corollary~\ref{cor:G1mis}, we see that the expected size of the independent set produced by $\A_d$ in any $G_1'$ has to be less than $n_Q/\alpha_d$, and hence it fails to find an $\alpha_d$-approximation in expectation in some $G_2'$. This concludes the proof of Theorem~\ref{thm:mislb}.

We point out that similar ideas can be used to prove a lower bound for the maximum cut problem. For example, in $d$-regular triangle-free graphs, it is possible to find a factor $1/(1/2 + 1/O(\sqrt{d}))$ approximation of the maximum cut in \emph{constant} time with randomized $\local$ model algorithms \cite{DBLP:journals/combinatorics/HirvonenRSS17}. We can show that switching to the $\supported$ model does not help: it is not possible to find a factor $1/(1/2 + 1/o(\sqrt{d}))$ approximation of the maximum cut in \emph{sublogarithmic} time with randomized $\supported$ algorithms.


\section{Link failures: The Passive \supported{} Model}\label{sec:PASSIVE}
So far we assumed that all edges in the support topology $H$ can be used for communication, no matter the input subgraph $G \subseteq H$.
However, if edges physically fail, then this assumption is no longer viable.
We thus introduce the \emph{passive} \supported{} model, where after the precomputation phase, communication is restricted to the input graph $G$.

For simplicity of presentation, we only consider edge failures\footnote{A node failure can be simulated by failing all incident edges.}
 and assume the communication graph is $G=(V,E)$.
Note that unlike most prior work, we do not restrict ourselves to a few failures, but allow \emph{any} subgraph $G$ of $H$.
We next provide a brief overview of our results for the passive model.

\paragraph{General graphs are problematic.}
From a very general point of view, precomputation on general graphs does not help much in the passive \supported{} model.
For example, if $H$ is a complete support graph, then it seems that only few meaningful information can be prepared against adversarial failures, such as an upper bound on the network size respectively obtaining a superset of the ID-space; no meaningful topological information about the structure of $G$ is known ahead of time.
We will formalize this intuition in Section~\ref{subsec:passive-sim}.

\paragraph{Restricted graph families are useful.}
This unsatisfying situation changes when we consider graph properties that are retained under edge deletions.
To give a prominent example, a planar graph remains planar, no matter what subgraph $G$ is selected.
Similarly, the genus or chromatic number of a graph only becomes smaller.
We will show in Section~\ref{subsec:planar-genus} how to speed up the runtime of some algorithms in these restricted graph families, beginning with planar graphs.

\paragraph{Weaker $\slocal$ simulations.}
Due to the possible physical edge failures, our $\slocal$-simulation from Section \ref{sec:SLOCAL} can no longer be applied. In particular on dense graphs, it relied extensively on support edges not contained in the input graph, but those edges are no longer available for communication purposes.
Notwithstanding, we can still provide a weaker simulation of the $\slocal$ model in Section~\ref{subsec:passive-slocal}.

\subsection{Simulation in the \local{} Model}\label{subsec:passive-sim}

We begin by showing that, under certain assumptions, the passive \supported{} model can be simulated in the deterministic \local{} model.
When a passive \supported{} algorithm runs in polylogarithmic time, the overhead is only a constant factor in the number of rounds, illustrating that in general the power of preprocessing in the passive case is rather limited.

\begin{theorem} \label{thm:local-passive-simulation}
  Consider the \local{} model with identifiers in $\{ 1,2,\dots,n \}$. Let $\Pi$ be a graph problem such that a feasible solution for each connected component is a feasible solution for the whole graph. Let $A$ be $T(n)$-time algorithm for $\Pi$ in the \emph{passive} \supported{} model that works for a support with maximum degree $\Delta = O(n^{1/k})$. Then $A$ can be simulated in time $O(T(n^k))$ in the \local{} model.
\end{theorem}

\begin{proof}
 Let $A$ be a $T(n)$-time algorithm for a problem $\Pi$. The nodes of $G$ agree on a virtual support $H$ as follows: $H$ consists of a clique $K_n$ on the real vertices in $V(G)$, and of $n^k - n$ virtual nodes. Each real node is connected to one virtual node, and the remaining virtual nodes form a connected graph of maximum degree $\Delta(H) = n$. Since all real nodes $v$ also have $\deg_H(v) = n$, we have that $\Delta(H) = n = |V(H)|^{1/k}$.

  Next nodes start simulating $A$ as if it was run on $H$, with the edges failing in a way that produces the observed input graph $G$. In particular, the edges between the real and the virtual nodes are assumed to have failed. Since $H$ contains the complete graph $K_n$ as its subgraph, all observable graphs $G$ can be formed from the support $H$ by edge deletions.

  Since the support $H$ has size $n^{k}$, we have that $A$ runs on $H$ (and thus $G$) in time $T(n^{k})$. If $T(n) = O(\poly \log n)$, then $T(n^{k}) = O(\poly \log n)$. Since $A$ produces a feasible output on $H$, it produces a feasible output on $H$ restricted to $G$. By assumption this is a feasible output on $G$.
\end{proof}

Note that the above result covers in particular the class of locally checkable labelling problems. The assumption that the nodes have names from $\{1,2,\dots,n\}$ is not usually useful for algorithm design in the \local{} model.

We can also consider optimization problems. We consider problems for which the size of the solution is the sum of the sizes over all connected components.

\begin{corollary} \label{cor:approximation-simulation}
  Let $A$ be a $T(n)$-time algorithm for an optimization problem $\Pi$ in the passive \supported{} model that produces a solution of size at least $\alpha n$ when $\Delta(H) = O(n^{1/k})$. Then there exists a \local{} algorithm that produces a solution of size $\alpha n$ in time $T(n^k)$.
\end{corollary}

\begin{proof}
  Without loss of generality we consider a problem where we want to \emph{maximize} the target function $f$. Let $N = n^k$. Let $G_0$ be the input graph on $n$ nodes in the \local{} model, and let $H$ denote the (virtual) support that consists of $N$ cliques of size $n$, denoted by $H_0, H_1, H_2, \dots, H_{N-1}$, connected by some subset of edges, at most one per node.

  We construct the input graph $G$ as follows, shuffling the identifiers of the graph as we go. First, remove all edges between two different cliques $H_i$ and $H_j$. Then, let $G_1$ equal the subgraph of $H_1$ and $\varphi_1$ the identifier setting from $[N]$ to $V(G_1)$ such that the value $f(G_1)$ is minimized over all subgraphs and all identifier settings, when $A$ is run on $G_1$. Let $\varphi(G_1)$ denote the set of identifiers used on $G_1$. Again, let $G_2$ be the subgraph of $H_2$ with the identifier setting $\varphi_2$ from $[N] \setminus \varphi(H_1)$ that minimizes $f(H_2)$ over all subgraphs and all identifier settings. We proceed in this manner until all $H_i$ have been dealt with, except for $H_0$.

  There is a subset $S = \{s_1, s_2, \dots, s_n \} \subseteq [N]$ of identifiers left. We choose the mapping $\varphi_0: i \mapsto s_i$ for the identifiers of $G_0$. This can be done consistently by all nodes of $G_0$ since the graphs $H_1, H_2, \dots, H_{N-1}$ can be constructed with the knowledge of $N$ and $A$. Now, since regardless of the topology of $G_0$, it is a subgraph of each $H_i$ and the set $S$ was under consideration when assigning each $\varphi_i$, we must have that $f(G_0) \geq f(G_{N-1}) \geq f(G_{N-2}) \geq \cdots \geq f(G_1)$. In addition, since we have by assumption that $\sum_{i=0}^{N-1} f(G_i) \geq \alpha N$, we must have that $f(G_0) \geq \alpha n$.
\end{proof}

\subsection{Breaking Locality Lower Bounds}\label{subsec:planar-genus}
\paragraph{Planar graphs.} Maximum independent set, maximum matching, and minimum dominating set are considered to be classic networking problems.
Already on planar graphs, a $(1 + \delta)$-approximation is impossible to compute in constant time for any of the three problems~\cite{DBLP:conf/wdag/CzygrinowHW08}.
However, we can adapt the algorithms of Czygrinow et al.~\cite{DBLP:conf/wdag/CzygrinowHW08} to break these locality lower bounds in the passive \supported{} model:

\begin{theorem}\label{thm:passive-supported-planar}
Let $H$ be a planar support graph and fix any $\delta >0$. In the passive \supported{} model, a $(1 + \delta)$-approximation can be computed in constant time on any input graph $G \subseteq H$ for the following three problems: $1)$ maximum independent set, $2)$ maximum matching, and $3)$ minimum dominating set.
\end{theorem}

\begin{proof}
We study the algorithm construction Czygrinow et al.~\cite{DBLP:conf/wdag/CzygrinowHW08}, which solves the three problems in non-constant time, and show how it can be adapted to the passive \supported{} model.
It is based on finding a weight-appropriate pseudo-forest\footnote{A directed graph with maximum out-degree of one.} in constant time, which in turn is 3-colored. The coloring allows to find so-called heavy stars in constant time, which are used as the base of a constant time clustering algorithm.
Due to the heavy stars having a constant diameter, the approximation in~\cite{DBLP:conf/wdag/CzygrinowHW08} also runs in constant time.

Hence, the only non-constant time component is to find a 3-coloring of the pseudo-forest.
As planar graphs have a chromatic number of 4, we can precompute such a 4-coloring in the passive \supported{} model, i.e., it remains to go from 4 to 3 colors in constant time.
To this end, we can make use of~\cite[Algorithm \emph{DPGreedy}]{Suomel2016}, which in pseudo-forests reduces an $x$-coloring to an $x-1$-coloring in two rounds, assuming that $x\geq 4$, which concludes the proof construction.
\end{proof}

\paragraph{Extensions to bounded genus graphs.}
Amiri et al.~\cite{DBLP:conf/podc/AmiriSS16,disc16ba,DBLP:journals/corr/AmiriSS17} show how to adapt the technique of  Czygrinow et al.~\cite{DBLP:conf/wdag/CzygrinowHW08} to graphs of bounded genus\footnote{Graphs of genus $0$ are the class of planar graphs.} $g \in O(1)$ for the minimum dominating set problem.
By careful analysis of their work, as in the proof of Theorem~\ref{thm:passive-supported-planar}, it only remains to provide a 3-coloring of the resulting pseudo-forests in constant time.
As a genus $g$ graph can be colored with $O(\sqrt{g})$ colors~\cite{Ringel438}, we can obtain such a 3-coloring in $O(\sqrt{g})$ time as well.
\begin{corollary}
Let $H$ be a support graph of constant genus $g$ and fix any $\delta >0$. In the passive \supported{} model, a $(1 + \delta)$-approximation can be computed in constant time on any input graph $G \subseteq H$ for the minimum dominating set problem.
\end{corollary}

\paragraph{Locally checkable labellings.}
We briefly re-visit graphs with some constant maximum degree $\Delta$, to investigate \lcl{}s in the passive $\supported$ model.
The bounded degree property allows us to efficiently precompute colorings, which we will also investigate in the next Section~\ref{subsec:passive-slocal}.
In fact, when we check the proof of our $\lcl$-Theorem~\ref{thm:lcl-collapse}, we just made use of the precomputed coloring, not of the (now failed) additional support edges.
As thus, symmetry breaking problems also collapse to constant time in the passive $\supported$ model:

\begin{corollary} \label{cor:lcl-collapse-passive}
  Every \lcl{} in \localp{o(\log n)} can be solved in time $O(1)$ in the passive \supported{} model.
\end{corollary}

\subsection{An \slocal{} Simulation in the Passive \supported{} Model}\label{subsec:passive-slocal}

As described earlier in this section, we can no longer rely on our efficient $\slocal$ simulation from Section~\ref{sec:SLOCAL} in the passive $\supported$ model.
Still, we can create another kind of dependency chain which also works in the passive model.

\begin{theorem}
  An \slocal{} algorithm with locality $t$ can be simulated in the passive \supported{} model in time $\Delta^{O(t)}$.
\end{theorem}

\begin{proof}The following proof is conceptually similar to the proof of Theorem~\ref{thm:slocal-sim}, but will rely on a dependency chain that is given by a coloring hierarchy, i.e., in the simulation, nodes of color 1 execute first, then color 2 etc.
When creating the coloring, we need to ensure that the resulting execution performs as a global $\slocal$ one.
To this end, for every $t$-neighborhood, it must hold that at most one node executes its actions at the same time:
then, as the locality $t$ is known, the executing nodes can gather the states of all nodes of smaller colors according to the given locality.
In other words, any two nodes within distance $2t$ of each other need to have distinct colors, which is satisfied by a distance-$(2+1)$ coloring.
As the maximum degree of the support graph being $\Delta$, such a coloring can be performed with $\Delta^{O(t)}$ colors.
The execution time for a single color class is $O(t)$, resulting in  $O(t)\Delta^{O(t)}$ rounds, which can be simplified to $\Delta^{O(t)}$ in the big-O notation for $\Delta >1$.
\end{proof}

\section{Related Work}\label{sec:relwork}

Decentralized and local network algorithms
have been studied for almost three decades
already~\cite{naor95what}, often with applications
in ad-hoc networks in mind.
More recently, decentralized approaches are also
discussed intensively in the context of software-defined
networks~\cite{Nguyen:2017:DCU:3050220.3050224,berde2014onos,yeganeh2013scalability,hassas2012kandoo,canini2015distributed}, which serve us as a case study
in this paper, and also motivated Schmid and Suomela~\cite{hotsdn13} to consider the $\supported$ model.

To the best of our knowledge, this paper
is the first to systematically explore the
novel opportunities and limitations introduced
by enhancing scalable network algorithms with preprocessing.
That said, there are several interesting
results which are relevant in this context as well.

The congested clique~\cite{DBLP:conf/spaa/LotkerPPP03} can be considered a special case of the $\supported$ model, where the support graph $H$ is the complete graph $K_n$.
However, only speedup without preprocessing is investigated, and communication is restricted to messages of logarithmic size, i.e., the so-called $\congest$~\cite{peleg00distributed} model.
Korhonen et al.\ \cite[\S7]{DBLP:conf/opodis/KorhonenR17} investigate the $\supported$ $\congest$ model (with support graphs of bounded degeneracy) and show how preprocessing can be leveraged for faster subgraph detection in sparse graphs.

\section{Conclusion}\label{sec:conclusion}

Scalability is a key challenge faced by the
quickly growing communication networks.
In this paper, we initiated the study of
how enhancing local and decentralized algorithms
with preprocessing (as it is often easily possible
in modern networks) can help to further improve
efficiency and scalability of such networks.
We presented several positive results on how
preprocessing can indeed be exploited, but also
pointed out limitations.

We understand our work as a first step and
believe that it opens many interesting questions
for future research. In particular, there
exist several fundamental algorithmic problems
for which the usefulness of preprocessing still
needs to be explored. Furthermore, it will also
be interesting to better understand the relationship
between the opportunities introduced by
supported models and the opportunities introduced by
randomization.

\section*{Acknowledgements}

This work was supported in part by the Academy of Finland, Grant 285721, and the Ulla Tuominen Foundation. Part of the work was conducted while JH was affiliated with the Institut de Recherche en Informatique Fondamentale and the University of Freiburg.

\urlstyle{same}

\bibliographystyle{plainnat}
\bibliography{sdn-locality}
\end{document}